\documentclass[12pt]{article}

\usepackage[latin1]{inputenc}
\usepackage{amssymb,amsmath,stmaryrd}
\usepackage{setspace}
\usepackage[usenames]{color}
\usepackage{oldgerm}
\usepackage{amsthm}

\newtheorem{theorem}{Theorem}
\newtheorem{lem}{Lemma}
\newtheorem{cor}{Corollary}
\newtheorem{rem}{Remark}
\newtheorem{prop}{Proposition}
\newtheorem{definition}{Definition}

\newcommand{\myparagraph}[1]{{\noindent\bf{#1}~~}}
\newcommand{\proofbox}{\hspace*{\fill} $\Box$}

\newcommand{\thmref}[1]{Theorem~\ref{thm:#1}}

\newcommand{\lemref}[1]{Lemma~\ref{lem:#1}}
\newcommand{\lemrefs}[2]{Lemmas~\ref{lem:#1} and~\ref{lem:#2}}

\newcommand{\corref}[1]{Corollary~\ref{cor:#1}}
\newcommand{\propref}[1]{Proposition~\ref{prop:#1}}

\newcommand{\secref}[1]{Section~\ref{sec:#1}}

\newcommand{\eq}[1]{equation~\eqref{eq:#1}}
\newcommand{\eqs}[2]{equations~\eqref{eq:#1} and~\eqref{eq:#2}}

\newcommand{\fn}[2]{\footnote{#2\label{fn:#1}}}
\newcommand*{\fnmark}[1]{\textsuperscript{\normalfont\ref{fn:#1}}}

\newcommand{\RR}{\mathbb R}
\newcommand{\ZZ}{\mathbb Z}
\newcommand{\BBB}{\mathfrak B}

\newcommand{\FFF}{\mathfrak F}
\newcommand{\GGG}{\mathfrak G}

\newcommand{\LLL}{\mathfrak L}

\newcommand{\TTT}{\mathfrak T}

\newcommand{\brr}[2]{\overline{(#1,#2)}}
\newcommand{\br}[1]{\overline{(#1,L)}}
\newcommand{\U}{IU}
\newcommand{\x}{\overline{x}}

\newcommand{\bx}{\overline{(x,L)}}

\newcommand{\ol}{\overline}
\newcommand{\oo}{{\bf{o}}}
\newcommand{\0}{{\bf{0}}}

\newcommand{\ra}{\rightarrow}

\newcommand{\menge}[2]{\big\{{#1}\ \big|\ {#2}\big\}}
\newcommand{\ie}{i.\,e.\ }

\newcommand{\al}{\alpha}
\newcommand{\g}{\gamma}

\DeclareMathOperator{\supp}{supp}

\begin{document}
\title{{Towards Minimizing $k$-Submodular Functions}\thanks{A short version of this paper has appeared in \cite{ISCO}.}}
\author{Anna Huber\\
University of Derby,
Kedleston Road,
Derby
DE22 1GB,
UK
 \and Vladimir Kolmogorov\\
 Institute of Science and Technology Austria\\
Am Campus 1, 3400 Klosterneuburg, Austria
}
\date{}


\maketitle

\begin{abstract}
In this paper we
investigate $k$-submodular functions. This natural family of discrete functions includes submodular and bisubmodular functions as the special cases $k=1$ and $k=2$ respectively.

In particular we generalize the known Min-Max-Theorem for submodular and bisubmodular functions. This theorem asserts that the minimum of the (bi)submodular function can be found by solving a maximization problem over a (bi)submodular polyhedron. We define and investigate a $k$-submodular polyhedron and prove a Min-Max-Theorem for $k$-submodular functions.
\end{abstract}

\section{Introduction}

A key task in combinatorial optimization is the minimization of discrete functions. One important example are submodular functions. They are a fundamental concept in combinatorial optimization \cite{Edmonds,Schrijver}, and they have numerous applications elsewhere, see \cite{Frank,FujiBook,Schrijver}.

\myparagraph{Submodular functions}are originally defined on the power set of a set. Specifically, a real-valued function $f$ is called \emph{submodular} if it satisfies $f(T \cap U) + f(T \cup U) \leq f(T) + f(U)$ for all subsets $T, U$. The problem of minimizing a given submodular function is one of the most important tractable optimization problems \cite{FujiBook,Schrijver}. Its importance is comparable to minimizing convex functions in the continuous case, see the correspondence between submodular and convex functions provided by Lov\'asz \cite{Lovasz}. Because of this, submodularity is also called discrete convexity.
On the way to polynomial minimization algorithms, a structural theory of submodular functions has been developed, see \cite{FujiBook,Iwata,McCormick}.
In particular a submodular polyhedron was defined and the classical Min-Max-Theorem by Edmonds asserts that the minimum of the submodular function can be found by maximizing the $L^1$-norm over the negative part of this polyhedron, see \cite{Edmonds}. The first polynomial-time algorithm was based on the ellipsoid method \cite{Ellipsoid}, further, combinatorial, strongly polynomial algorithms are based on the Min-Max-Theorem \cite{IwataFleischerFuji,IwataOrlin,SchrijverAlgo,Orlin}.

Following a question of Lov\'asz \cite{Lovasz}, submodularity has been generalized to {\bf bisubmodularity}.
Bisubmodular functions were introduced under the name directed submodular functions in \cite{Qi}. Independently, they have been introduced as rank functions of pseudo\-matroids in \cite{Bouchet,ChandrasekaranKabadi,KabadiChandrasekaran}. Bisubmodular functions and
their generalizations have also been considered in \cite{BouchetCunningham,FujiBook,Nakamura}.

It has been shown that some structural results on submodular functions can be generalized to bisubmodular
functions. In particular, for every bisubmodular function a polyhedron is defined, and a Min-Max-Theorem tells us that the minimum of the bisubmodular function can be obtained by maximizing a linear function over this polyhedron, see \cite{CunGreen,FujiMinMax}. Using this Min-Max-Theorem, weakly polynomial and later strongly polynomial algorithms to minimize bisubmodular functions have been obtained \cite{FujiAlg,McCormickFuji}.

\myparagraph{This work: $k$-submodular functions.}
In this paper we investigate {\em $k$-submodular functions}, which generalize submodular and bisubmodular functions
in a natural way. They are defined on the product of trees of height one (i.e.\ stars) with $k$ leaves, $k \geq1$.
Submodular and bisubmodular functions are included in our setting as the special cases $k=1$ and $k=2$ respectively.
There is also a relation to {\em multimatroids} introduced in~\cite{Bouchet:I,Bouchet:II,Bouchet:III}: as we 
show in this paper, rank functions of multimatroids (or more precisely of $k$-matroids) are $k$-submodular.

$k$-submodular functions are special cases of {\em (strongly) tree-submodular functions} introduced
in~\cite{Kolmogorov}. The crucial question left open in~\cite{Kolmogorov} is whether $k$-submodular functions can be minimized efficiently.
As shown in~\cite{Kolmogorov}, a positive answer would yield tractability of tree-submodular function minimization for all trees.

The first approach for minimizing (bi)submodular functions, the reduction to convex optimization via the Lov\'asz extension, does not seem to work in our setting: In the submodular case, we obtain a convex optimization problem in the positive ortant of the Euclidean space. In the bisubmodular case, we have the two signs $+$ and $-$ and so are no longer in the positive ortant but in the whole Euclidean space. We still get a convex optimization problem. In the case $k\geq 3$, we now have more than two possible labels with each positive number, so we cannot represent them as pairwise opposite signs any more. We could call these labels ``colours'' and would get an optimization problem in a ``coloured space'', but to our knowledge nothing is known about ``coloured convexity''.

We thus have to start by investigating the structure of $k$-submodular functions.
We will generalize some notions and results from (bi)submodular functions to $k$-submodular functions.
{\bf Our contributions} are as follows.
First, we prove a generalization of the Min-Max-Theorem (\secref{properties}); this theorem has been the foundation of most (bi)submodular function minimization algorithms. 
Second, we introduce and analyze the polyhedron associated with $k$-submodular functions (\secref{picture}).
In \secref{multimatroids} we discuss the relationship between $k$-submodular functions and multimatroids.
Finally, in \secref{disc} we describe some difficulties regarding the generalization of (bi)submodular function minimization algorithms to the case $k\ge 3$.

\myparagraph{Related work: VCSPs and multimorphisms.}
There is a strong connection between submodular function minimization and {\em Valued Constraint Satisfaction Problems} (VCSPs).
The VCSP is a general combinatorial framework that allows to study complexity of certain classes of optimization problems.
In this framework one is given a {\em language} over a fixed finite domain $\TTT$, i.e.\ a collection of cost functions $f:\TTT^m\ra\RR$, where the arity $m$ may depend on $f$.
We can now pose the following question: what is the complexity of minimizing functions that can be expressed by summing functions with overlapping sets of variables from the language?
Studying the complexity for different languages has been an active research topic. For some cases, see e.g.~\cite{softCSP,Deineko,JonssonKuivinenThapper,KolmogorovZivny,Takhanov}, researchers have established dichotomy theorems
of the following form: if all functions from the language admit certain {\em multimorphisms} then the language is tractable, otherwise it is NP-hard.
The notion of multimorphisms is thus central in this line of research.

A (binary) multimorphism is a pair of operations $\sqcap,\sqcup:\TTT\times\TTT\ra\TTT$.
We denote the componentwise operations on $\TTT^n$ also by $\sqcap$ and $\sqcup$.
A function $f:\TTT^n\ra\RR$ is said to {\em admit the multimorphism} $\langle\sqcap,\sqcup\rangle$, if $f(T \sqcap U) + f(T \sqcup U) \leq f(T) + f(U)$ for all $T,U\in\TTT^n$.
Clearly, submodular and bisubmodular functions correspond to particular choices of $\langle\sqcap,\sqcup\rangle$, and so do $k$-submodular functions.

A rather general result on the tractability of VCSP-languages admitting certain multimorphisms has been shown in \cite{Raghavendra,ThapperZivny}. It includes $k$-sub\-modu\-lar languages and thus implies the tractability of the minimization problem in the {\bf VCSP model}, i.e.\ when the function to be minimized is given as a sum of local $k$-submodular functions of bounded size. However, the tractability of the minimization problem  of general $k$-submodular functions in the {\bf oracle model} remains open.
In the latter model the algorithm is allowed to access a function $f:\TTT^n\ra\RR$ only by querying for the value of $f(T)$ for any $T\in\TTT^n$.

The general problem of {\em multimorphism function minimization} is raised in \cite{JonssonKuivinenThapper}. 
There are many examples of ``tractable'' multimorphisms.
One of them is the pair of ``meet'' and ``join'' operations on distributive lattices \cite{Topkis,TopkisBook}
and some non-distributive lattices \cite{Krokhin,Kuivinen}.
In \cite{JonssonKuivinenThapper} a new multimorphism is introduced and used to characterize maximum constraint satisfaction problems on a four-element domain.
See also
\cite{TournamentPairMultimorphisms}
for another example of a multimorphism characterizing tractable optimization problems.

Multimorphisms that have proved to be important in this context often seem to be submodular-like. It thus seems promising to study the multimorphism function minimization problem for multimorphisms generalizing submodularity, like $k$-submodularity.

\myparagraph{Recent related work.}
Very recently,  Fujishige and Tanigawa~\cite{FujishigeTanigawa} proved a weaker version of the Min-Max-Theorem, \thmref{main}. We describe how their result follows from ours in \secref{picture}.

In the conference version of this paper, \cite{ISCO}, we claimed a characterization of extreme points of a polyhedron associated with $k$-submodular functions.
Unfortunately, this characterization contained a mistake. We thank Fujishige and Tanigawa for pointing it out. In this version we completely omit Lemma 7 and Section 4.1 from \cite{ISCO}, as they rely on Lemma 6 which is false. See~\cite{FujishigeTanigawa} for an alternative description of the extreme points.

Finally, we would like to mention the work~\cite{Potts} which presented an efficient algorithm
for minimizing a subclass of $k$-submodular functions together with an application in computer vision.

\section{Definitions and Notations}\label{sec:not}
Let $k \in \ZZ_{\geq1}$ and let $\TTT$ be a tree of height $1$ on $k+1$ vertices, \ie a star rooted at the non-leaf. By $\LLL$ we will denote the set of leaves, by $\oo$ the root. Note that $|\LLL| = k$.
We define the operations ``\emph{intersection}'' $\sqcap$ and ``\emph{union}'' $\sqcup$ on $\TTT$ as being idempotent (\ie $t \sqcap t := t=: t \sqcup t$ for every $t \in \TTT$) and for two distinct leaves $a, b \in \LLL$ as follows.
\begin{eqnarray*}
a \sqcap b := &\oo& =: a \sqcup b, \\
a \sqcap \oo := &\oo& =: \oo \sqcap a,\  \mbox{and}\\
a \sqcup \oo := &a& =: \oo \sqcup a.
\end{eqnarray*}
Let $n \in \ZZ_{\geq1}$.
On $\TTT^n$ intersection $\sqcap$ and union $\sqcup$ are defined componentwise. We write $\0:=(\oo)_{i=1}^{n}$.

A function $f: \TTT^n \ra \RR$ is called \emph{$k$-modular}, 
if for all $T, U \in \TTT^n$
\begin{equation*}
 f(T \sqcap U) + f(T \sqcup U) = f(T) + f(U),
\end{equation*}
\emph{$k$-submodular}, 
if for all $T, U \in \TTT^n$
\begin{equation}\label{eq:submodular}
 f(T \sqcap U) + f(T \sqcup U) \leq f(T) + f(U),
\end{equation}
 and \emph{$k$-supermodular}, 
if for all $T, U \in \TTT^n$
$$ f(T \sqcap U) + f(T \sqcup U) \geq f(T) + f(U).$$
Our definitions include (sub-/super-)modular set functions as the case $k=1$. The functions we get in the case $k=2$ are the bi(sub-/super-)modular functions introduced under the name directed (sub-/super-)modular functions in \cite{Qi}.

\section{Min-Max-Theorem for $k$-Submodular Functions}\label{sec:properties}\label{sec:proof}

In this section we will state and prove our first result, \thmref{main}, which says that we can minimize a $k$-submodular function by maximizing the $L^1$-norm over an appropriate subset of the Euclidean space. This result is intended to play the same role in the $k$-submodular context which the classical Min-Max-Theorem of Edmonds \cite{Edmonds} plays in the ordinary submodular minimization context.

For the remainder of this section we will assume $k\geq 2$. The case $k=1$ can be easily included with only a minor technical change in notation.

\subsection{The Min-Max-Theorem}\label{sec:minmax}
For any $x = (x_i)_{i=1}^{n} \in \RR_{\geq 0}^n$, let $\|x\| := \sum\limits_{i=1}^{n}x_i$ denote the $L^1$-norm.

For any $(x,L) \in \RR_{\geq 0}^n \times \LLL^n$, let $\bx : \TTT^n \ra \RR$ be defined as follows.\\
For every $i \in [n] := \{1, \dots , n\}$, let $\bx_i : \TTT \ra \RR $ be defined through $$\bx_i(\oo) := 0,\ \bx_i(L_i) := x_i,\ \mbox{and}\ \ \bx_i(\ell) := - x_i\mbox{ for}\ \ \ell \in \LLL \setminus \{L_i\},$$ 
where $L = (L_i)_{i=1}^{n}$.
For every $T \in \TTT^n$, let
$$\bx(T) := \sum_{i=1}^{n}\bx_i(T_i).$$

\begin{prop}\label{prop:supp}
For $(x,L) \in \RR_{\geq 0}^n \times \LLL^n$, the function $\bx$ is $k$-supermodular.
\end{prop}
\begin{proof}
It follows from the definition that for every $i \in [n]$ the function $\bx_i$ is $k$-supermodular, which carries over to the sum.
\end{proof}

For any function $f: \TTT^n \ra \RR$, we define
$$U(f) := \menge{  (x,L) \in \RR_{\geq 0}^n \times \LLL^n}{\forall\  T \in \TTT^n\ \ \ \bx(T) \leq f(T)}.$$
In the special cases $k=1$ and $k=2$, the set $U(f)$ corresponds to the submodular polyhedron\fn{1}{We can extend our notation to the case $k = 1$ by defining $\bx_i(\ell) := - x_i$ for $\{\ell\} = \LLL$.
If we do so, the set $-U(f)$ is the negative orthant of the submodular polyhedron as in \cite{Edmonds}.
Only this orthant proves to be relevant in the classical Min-Max-Theorem for submodular functions \cite{Edmonds}, 
see \fnmark{2}.
}
 and the bisubmodular polyhedron, respectively.
In the case $k=2$ note that we can consider the two leaves just as the signs $+$ and $-$, so $\RR_{\geq 0} \times \LLL$ corresponds to $\RR$. Instead of $\bx$ for every $(x,L) \in \RR_{\geq 0}^n \times \LLL^n$, we have $\x : \{\oo, +, -\}^n \ra \RR$ for every $x \in \RR^n$ with $\x(T) = \sum\limits_{T_i = +}x_i - \sum\limits_{T_i = -}x_i$
, and $U(f)$ just reads $\menge{x \in \RR^n}{\forall\  T \in \TTT^n\ \ \ \x(T) \leq f(T)}$. This is the usual bisubmodular polyhedron as in \cite{CunGreen,FujiMinMax}. For $k\geq 3$, despite being a natural generalization of the (bi)submodular polyhedra, the set $U(f)$ is not necessarily a polyhedron anymore. For an embedding of $U(f)$ into a polyhedron in a higher-dimensional Euclidean space see \secref{picture}, however, $U(f)$ is not necessarily convex. Nevertheless, it turns out to be the set of all \emph{unified} vectors in a $k$-submodular polyhedron, see \secref{picture} for the details. Unified vectors play an important role in the tractability result in \cite{Kuivinen}.

We have the following main theorem.
\begin{theorem}\label{thm:main}
Let $f: \TTT^n \ra \RR$ be $k$-submodular, $f(\0) = 0$. Then
\begin{equation}\label{eq:Main}
 \min_{T \in \TTT^n}f(T) = \max_{(x,L) \in U(f)} - \|x\|.
\end{equation}
\end{theorem}
This theorem is a generalization of the classical Min-Max-Theorem of Edmonds for submodular functions\fn{2}{If we extend our notation as in \fnmark{1}, \thmref{main} for $k = 1$ is exactly the Min-Max-Theorem of Edmonds \cite{Edmonds}.},
 and the bisubmodular Min-Max-Theorem \cite{CunGreen,FujiMinMax}. 
The bisubmodular Min-Max-Theorem reads as follows in our notation: For any $2$-submodular function $f: \TTT^n \ra \RR$ with $f(\0) = 0$ and for any $x^0 \in \RR^n$, it is shown in \cite{FujiMinMax} that
$$ \min_{T \in \TTT^n}f(T)-\ol{x^0}(T) = \max_{x \in U(f)} \sum_{i=1}^n - |x_i-x_i^0|.$$
We get this by applying \thmref{main} to the function $f-\ol{x^0}$. This function is $2$-submodular, as in the particular case $k=2$ the function $\ol{x^0}$ is $2$-modular, and it fulfills $(f-\ol{x^0})(\0) = 0$.

Recently, in~\cite{FujishigeTanigawa}, a weaker version of \thmref{main} was obtained. We describe how their result follows from ours in \secref{picture}. 

The remainder of this section is devoted to the proof of \thmref{main}. 
We assume throughout that
$f: \TTT^n \ra \RR$ is a $k$-submodular function with $f(\0) = 0$.

\subsection{Properties of $U(f)$}

In this section we collect some properties of the set $U(f)$ which we will need for the proof in \secref{subproof}.
The following two lemmas are inspired by the bisubmodular case as treated in \cite{FujiBook}. They provide a reduction to the submodular case in the following sense. Recall that the case $k=1$ is the submodular case. If we now have $k \geq 2$, we choose one leaf for every coordinate $i \in [n]$ and restrict our function to the trivial trees on the root and this leaf only.

Let $\leq$ denote the partial order on $\TTT$ such that $\oo \leq t$ for all $t \in \TTT$ and all leaves are pairwise incomparable. 
Let $\leq$ also denote the componentwise partial order on $\TTT^n$.
For every $K \in \LLL^n$,
let
$$2^K := \menge{T \in \TTT^n}{T \leq K}.$$
For $K \in \LLL^n$, we define
$$U_K(f) := \menge{  (x,L) \in \RR_{\geq 0}^n \times \LLL^n}{\forall\  T \in 2^K\ \ \ \bx(T) \leq f(T)}$$
and the \emph{base set} 
$$B_K(f) := \menge{  (x,L) \in U_K(f)}{\bx(K) = f(K)}.$$
In the next two lemmas we show properties of $B_K(f)$ and conclude with the non-emptyness of $U(f)$ in \corref{non-empty}.
\begin{lem}\label{lem:containment}
For every $K \in \LLL^n$, one has
$$B_K(f) \subseteq U(f) \subseteq U_K(f).$$
\end{lem}
\begin{proof}
 The second inclusion is clear by the definitions. To see the first inclusion, let $T \in \TTT^n$ and $(x,L) \in B_K(f)$. Then the $k$-supermodularity of $\bx$ and the $k$-submodularity of $f$ give 
\begin{eqnarray*}
f(T) - \bx(T) &=& f(T) - \bx(T) + f(K) - \bx(K) \\
&\geq & f(T \sqcap K) - \bx(T \sqcap K) + f(T \sqcup K) - \bx(T \sqcup K). 
\end{eqnarray*}
As $(x,L) \in U_K(f)$ and $T \sqcap K, T \sqcup K \in 2^K,$ the right hand side is greater or equal zero and so one has
$$f(T) \geq \bx(T).$$
\end{proof}

\begin{lem}\label{lem:B_K(f)}
For every $K \in \LLL^n$ the base set $B_K(f)$ is non-empty.
\end{lem}
\begin{proof}
For any $g: \TTT^n \ra \RR$ the restriction $g|_{2^K}$ gives a set function $g_{K}: 2^{[n]} \ra \RR$ in the canonical way.
The set function $f_{K}$, obtained from $f$ as shown above, is submodular.
By the known facts about submodular set functions,
the submodular base polyhedron $B(f_{K})$ is non-empty, see for example \cite{Iwata} or \cite{FujiBook}.
The lemma follows if we show that the function
$$B_K(f) \ra B(f_{K})$$
$$(x,L) \mapsto \bx_{K}$$
is surjective. To show this let $y \in B(f_{K})$ and denote $x := (|y_i|)_{i=1}^{n}$. Then for every choice of $L \in \LLL^n$ such that $L_i = K_i$ if $y_i > 0$ and $L_i \neq K_i$ if $y_i < 0$, we have $(x,L) \in B_K(f)$ and $\bx_{K} = y$.
\end{proof}

\begin{cor}\label{cor:non-empty}
$U(f)$ is non-empty.
\end{cor}
\begin{proof}
 Follows from \lemrefs{containment}{B_K(f)}.
\end{proof}

Let $(x,L) \in U(f)$ be fixed for the remainder of this section. We say that an element $T \in \TTT^n$ is \emph{$(x,L)$-tight} if $$\bx(T) = f(T)$$ holds, and we define $\FFF(x,L)$
as being the set of $(x,L)$-tight elements of $\TTT^n$. We have the following.

\begin{prop}\label{prop:closed}
The set $\FFF(x,L)$ is closed under $\sqcap$ and $\sqcup$, and the function $\bx\big|_{\FFF(x,L)}$ is $k$-modular.
\end{prop}
\begin{proof}
Let $T, U \in \FFF(x,L)$.
By the definition of $U(f)$ we have
\begin{equation}\label{eq:cap}
 \bx(T \sqcap U) \leq f(T \sqcap U)
\end{equation}
and 
\begin{equation}\label{eq:cup}
 \bx(T \sqcup U) \leq f(T \sqcup U).
\end{equation}
Together with the $k$-submodularity of $f$ and the $k$-supermodularity of $\bx$ this yields
\begin{eqnarray*}
 f(T) + f(U)
&=&  \bx(T) +  \bx(U)\\
&\leq&  \bx(T \sqcap U) +  \bx(T \sqcup U)\\
&\leq&  f(T \sqcap U) +  f(T \sqcup U)\\
&\leq&  f(T) + f(U).
\end{eqnarray*}
We thus have equality here as well as in \eqs{cap}{cup}.
\end{proof}
Let \emph{$\supp(x)$} denote the support of $x$.

The next lemma is the key lemma for the proof of our main theorem, as it essentially provides a reduction to the case $k=2$, which is the bisubmodular case.

Informally, for any coordinate $i \in \supp(x)$, if we call $L_i$ the ``positive'' leaf and $\LLL \setminus \{L_i\}$ the set of ``negative'' leaves, the lemma states that, for tight elements, at most one ``negative'' leaf is possible in this coordinate.
\begin{lem}\label{lem:unique}
 If $T, U \in \FFF(x,L)$ and $i \in \supp(x)$ are such that $T_i, U_i \in \LLL \setminus \{L_i\}$, then $T_i=U_i$.
\end{lem}
\begin{proof}
Let $T, U \in \FFF(x,L)$.

\propref{closed} and $k$-supermodularity of $\bx_i$ for every $i \in [n]$ yields
$$\bx_i(T_i)+\bx_i(U_i)=\bx_i(T_i\sqcup U_i)+\bx_i(T_i\sqcap U_i)$$
 for every $i \in [n]$
which is not possible if $i \in \supp(x)$ and $T_i, U_i \in \LLL \setminus \{L_i\}$ and $T_i\neq U_i$, as then the left hand side would be negative and the right hand side would be zero.
\end{proof}

Let 
$$S(x,L) := \menge{i \in \supp(x)}{\exists \ \ T \in \FFF(x,L):\ T_i \in  \LLL \setminus \{L_i\}}.$$ 
For every $i \in S(x,L)$, the leaf $T_i$ is unique by \lemref{unique}, independent of the chosen $T \in \FFF(x,L)$ with $T_i \in  \LLL \setminus \{L_i\}$. We denote it by $\ol{L}_i$ and define
$$N((x,L),i) := \bigsqcap \menge{T \in \FFF(x,L)}{T_i = \ol{L}_i}.$$
It is well-defined since the operation $\sqcap$ is associative.
The introduction of the ``negative'' leaf $\ol{L}_i$ is a core point in our proof. \lemref{unique} provides a partial reduction to the bisubmodular case in the following sense. For $(x,L)$-tight elements, in every coordinate $i \in \supp(x)$ we now have only the choice between the leaves $L_i$ and $\ol{L}_i$,  ``positive'' and ``negative'' leaf, as it would be in the bisubmodular case where we only have two leaves available.

\begin{lem}\label{lem:compliant1}
 Let $T \in \FFF(x,L)$ and $i \in S(x,L)$ such that $T_i \leq \ol{L}_i$.
If $j \in [n]$ such that $N((x,L),i)_j \in  \LLL$, then
$$T_j \leq N((x,L),i)_j.$$
\end{lem}
\begin{proof}
Define $T' := (T \sqcup N((x,L),i)) \sqcap N((x,L),i)$. As $T' \in \FFF(x,L)$ and $T'_i = \ol{L}_i$, the definition of $N((x,L),i)$ yields $T'_j = N((x,L),i)_j$. So $(T \sqcup N((x,L),i))_j = N((x,L),i)_j$ holds which yields $T_j \leq N((x,L),i)_j$.
\end{proof}

For $i \in [n]$ let $\chi_i$ denote the \emph{characteristic vector}, \ie  $(\chi_i)_i = 1$ and $(\chi_i)_j = 0$ for $j \in [n] \setminus \{i\}$.

\begin{lem}\label{lem:ij}
 Let $ i \in S(x,L)$ and $j \in \supp(x)$. If $N((x,L),i)_j = L_j$, then there is an $\al > 0$ such that $(x-\al(\chi_i + \chi_j),L) \in U(f).$
\end{lem}
\begin{proof}
 Assume that for all $\al > 0$ one has $(x-\al(\chi_i + \chi_j),L) \notin U(f).$ Then there is a $T \in \FFF(x,L)$ such that $\br{\chi_i + \chi_j}(T) < 0$. So either
\begin{enumerate}
 \item $T_i \in  \LLL \setminus \{L_i\}$ and $T_j \neq L_j$ or
 \item $T_j \in  \LLL \setminus \{L_j\}$ and $T_i =\oo$.
\end{enumerate}
Case 2. is a contradiction to \lemref{compliant1}. Case 1. yields
\begin{enumerate}
 \item $T_i = \ol{L}_i$ and $T_j \neq L_j$,
\end{enumerate}
which is a contradiction to $N((x,L),i)_j = L_j$ by the definition of $N((x,L),i)$.
\end{proof}

\begin{lem}\label{lem:leq}
If $S(x,L) = \supp(x)$
and the operation $\sqcup$ is associative on the set
$\menge{N((x,L),i)}{i \in S(x,L)}$, then
$$\min_{T \in \TTT^n}f(T) \leq - \|x\|.$$
\end{lem}
\begin{proof}
We will show that there is a $T \in \TTT^n$ with $f(T) = -\|x\|.$ As we have associativity we can define
$$T := \bigsqcup_{i \in \supp(x)}N((x,L),i).$$
We have $T \in \FFF(x,L)$ and $T|_{\supp(x)} = \ol{L}|_{\supp(x)}$, so
$$f(T) = \brr{x}{L}(T)
 = -\brr{x}{L}(L) = -\|x\|.$$
\end{proof}

\subsection{Proof of \thmref{main}}\label{sec:subproof}
We now have collected all the properties of $U(f)$
we will need for the proof of the main theorem.

\begin{proof}[Proof of \thmref{main}]
For any $T \in \TTT^n$ and $(x,L) \in U(f)$, one has by definition
$$f(T) \geq \bx(T) \geq -\|x\|,$$
so
$$\min_{T \in \TTT^n}f(T) \geq \max_{(x,L) \in U(f)} -\|x\|.$$
To show $\displaystyle \min_{T \in \TTT^n}f(T) \leq - \min_{(x,L) \in U(f)} \|x\| = \max_{(x,L) \in U(f)} -\|x\|,$ we choose a $(\hat{x},\hat{L}) \in U(f)$ with $\displaystyle \|\hat{x}\|= \min_{(x,L) \in U(f)} \|x\|$.

By \lemref{leq} it is sufficient to show that
one has $S(\hat{x},\hat{L}) = \supp(\hat{x})$
and the operation $\sqcup$ is associative on $\menge{N((\hat{x},\hat{L}),i)}{i \in S(\hat{x},\hat{L})}$.

By minimality of $(\hat{x},\hat{L})$ for all $i \in \supp(\hat{x})$, one has
$$ \forall\ \al > 0\ \ (\hat{x}-\al\chi_i,\hat{L}) \notin U(f).$$
This means that there is a $T \in \FFF(\hat{x},\hat{L})$ such that for all $\al \in\ ]0,x_i]$ one has $\ol{(\hat{x}-\al\chi_i)}(T) > f(T)$, which yields $T_i \in \LLL \setminus \{\hat{L}_i\}$,
so $i \in S(\hat{x},\hat{L})$. 

To prove the associativity it is sufficient to show that for all $i, j \in S(\hat{x},\hat{L})$ and $m \in [n]$ 
the $m$-th coordinates
$N((\hat{x},\hat{L}),i)_m$ and
$N((\hat{x},\hat{L}),j)_m$ cannot be distinct leaves. 
We cannot have $N((\hat{x},\hat{L}),i)_j = \hat{L}_j$ by the minimality of $(\hat{x},\hat{L})$ and \lemref{ij}, so as $N((\hat{x},\hat{L}),i) \in \FFF(\hat{x},\hat{L})$ \lemref{unique} gives
$N((\hat{x},\hat{L}),i)_j \leq \ol{\hat{L}}_j.$
If $N((\hat{x},\hat{L}),j)_m \in  \LLL$, by \lemref{compliant1} one has
$N((\hat{x},\hat{L}),i)_m \leq N((\hat{x},\hat{L}),j)_m.$
\end{proof}

\subsection{An Integer Minimizer}

In this section we will show the existence of an integer minimizer of $\|x\|$ in $U(f)$ 
if the function $f$ is integer.

Let $f: \TTT^n \ra \ZZ$ be $k$-submodular, and $f(\0) = 0$. Let
$$\U(f) := U(f) \sqcap (\ZZ_{\geq 0}^n \times \LLL^n).$$

\begin{cor}
$\U(f)$ is non-empty.
\end{cor}
\begin{proof}
Follows from \lemref{containment} as in the proof of \lemref{B_K(f)} together with the fact that
the submodular base polyhedron of an integer function has integer vertices, see for example \cite{Edmonds}, \cite{Iwata}, or \cite{FujiBook}. 
\end{proof}

The proof of the next lemma basically follows the proof of \thmref{main}, we just have to be a bit more careful with the minimality arguments.

\begin{lem}\label{lem:compliant3}
We have
\begin{equation*}
 \min_{T \in \TTT^n}f(T) = \max_{(x,L) \in \U(f)} - \|x\|.
\end{equation*}
\end{lem}
\begin{proof}
The inequality
$$\min_{T \in \TTT^n}f(T) \geq  \max_{(x,L) \in \U(f)} - \|x\|$$
follows from \thmref{main}.
Let $(\hat{x},\hat{L}) \in \U(f)$ with $\displaystyle \|\hat{x}\|= \min_{(x,L) \in \U(f)} \|x\|$.

By \lemref{leq} it is sufficient to show that
one has $S(\hat{x},\hat{L}) = \supp(\hat{x})$
and the operation $\sqcup$ is associative on $\menge{N((\hat{x},\hat{L}),i)}{i \in S(\hat{x},\hat{L})}$.
By minimality of $(\hat{x},\hat{L})$, for all $i \in \supp(\hat{x})$ one has
$$ \forall\ \al > 0\ \ (\hat{x}-\al\chi_i,\hat{L}) \notin \U(f).$$
This yields
\begin{equation*}
\forall\ \al > 0\ \ (\hat{x}-\al\chi_i,\hat{L}) \notin U(f), 
\end{equation*}
as the existence of an $\al > 0$ such that $(\hat{x}-\al\chi_i,\hat{L}) \in U(f)$ would yield by the integrality of $\hat{x}$ and $f$ the existence of an integer such $\al > 0$, which would mean $(\hat{x}-\al\chi_i,\hat{L}) \in \U(f)$.
As in \thmref{main} this yields $i \in S(\hat{x},\hat{L})$. 

To prove the associativity, as in \thmref{main}, we show that for all $i, j \in S(\hat{x},\hat{L})$ and $m \in [n]$ 
the $m$-th coordinates
$N((\hat{x},\hat{L}),i)_m$ and
$N((\hat{x},\hat{L}),j)_m$ cannot be distinct leaves. 

To prove that we cannot have $N((\hat{x},\hat{L}),i)_j = \hat{L}_j$, we have to do a bit more work than in \thmref{main}.
Minimality of $(\hat{x},\hat{L})$ gives us only the non-existence of an integer $\al > 0$ as in \lemref{ij}.
So if we assume $N((\hat{x},\hat{L}),i)_j = \hat{L}_j$, \lemref{ij} gives us by the integrality of $\hat{x}$ and $f$ and the minimality of $(\hat{x},\hat{L})$ that
\begin{equation}\label{eq:alij}
(\hat{x}-\tfrac{1}{2}(\chi_i + \chi_j),\hat{L}) \in U(f). 
\end{equation}
As $i \in S(\hat{x},\hat{L})$, there is a $S \in \FFF(\hat{x},\hat{L})$ with $S_i = \ol{\hat{L}}_i$. By \eq{alij} the only possibility for coordinate $j$ is $S_j = {\hat{L}}_j$.

Let $T \in \TTT^n$ with $T_i  = \ol{\hat{L}}_i$ and  $T_j  = \ol{\hat{L}}_j$ 
By \eq{alij} one has
 $\brr{\hat{x}}{\hat{L}}(T) + 1 \leq f(T)$,
which yields by minimality
 $\brr{\hat{x}}{\hat{L}}(T) + 1 = f(T)$. We thus have
$$f(S) + f(T)
 = \brr{\hat{x}}{\hat{L}}(S) + \brr{\hat{x}}{\hat{L}}(T) +1,$$
which yields by the submodularity of $f$ and the supermodularity of $\brr{\hat{x}}{\hat{L}}$
$$f(S \sqcap T) + f(S \sqcup T) \leq 
 \brr{\hat{x}}{\hat{L}}(S \sqcap T) + \brr{\hat{x}}{\hat{L}}(S \sqcup T) +1 .$$

We have $(S \sqcap T)_i = (S \sqcup T)_i = \ol{\hat{L}}_i$ and $(S \sqcap T)_j=(S \sqcup T)_j=\oo$ and so by \eq{alij} and integrality $\brr{\hat{x}}{\hat{L}}(S \sqcap T) + 1\leq f(S \sqcap T)$ and $\brr{\hat{x}}{\hat{L}}(S \sqcup T) + 1\leq f(S \sqcup T)$, a contradiction.

So we cannot have $N((\hat{x},\hat{L}),i)_j = \hat{L}_j$ and thus have
$N((\hat{x},\hat{L}),i)_j \leq \ol{\hat{L}}_j.$
If $N((\hat{x},\hat{L}),j)_m \in  \LLL$, by \lemref{compliant1} one has
$N((\hat{x},\hat{L}),i)_m \leq N((\hat{x},\hat{L}),j)_m.$
\end{proof}

\section{The $k$-Submodular Polyhedron}\label{sec:picture}
In this section we will generalize several notions from \secref{proof} to a higher-dimensional space in order to define a \emph{$k$-submodular polyhedron} $P(f)$, in ana\-logy to the polyhedra defined in the ordinary submodular case, see \cite{Edmonds}, and the bisubmodular case, see \cite{CunGreen,FujiMinMax}. We show how $U(f)$ can be embedded in $P(f)$ and investigate the properties of the polyhedron.

For any $x \in \RR^{n \times \LLL}$, we write $x = (x_{i\ell})_{i \in [n],\ \ell \in \LLL}$, and also $x_i = (x_{i\ell})_{\ell \in \LLL}$ for every $i \in [n]$. We define $\x : \TTT^n \ra \RR$ as follows. For every $i \in [n]$, let $\x_i : \TTT \ra \RR $ be defined through $\x_i(\oo) := 0$, and $\x_i(\ell) := x_{i\ell}$ for $\ell \in \LLL$. For every $T \in \TTT^n$ let

$$\x(T) := \sum_{i=1}^{n}\x_i(T_i).$$

For any $k$-submodular function $f: \TTT^n \ra \RR$ with $f(\0) = 0$, we define the polyhedron
\begin{align*}
P(f) := \big\{x \in \RR^{n \times \LLL}\ \big|\ &\forall\  T \in \TTT^n\ \ \ \x(T) \leq f(T)\ \ \ \ \mbox{and}\\ &\forall\ i \in [n] \ \ \forall\  \{\ell,p\} \in {\LLL \choose 2}\  \ \ x_{i\ell}+x_{ip} \leq 0\big\}. 
\end{align*}
For $k=1$ this is exactly the definition of a submodular polyhedron as in \cite{Edmonds}, and for $k=2$ this is 
a superset of the usual bisubmodular polyhedron
as introduced in
\cite{DunstanWelsh},
see also
\cite{BouchetCunningham,CunGreen,FujiMinMax}. If we write $\LLL = \{\ell,p\}$ we have
$$P(f) = \big\{x \in \RR^{n \times \LLL}\ \big|\ \forall\  T \in \TTT^n\ \ \ \x(T) \leq f(T)\ \ \mbox{and}\ \ \forall\ i \in [n]\  \ \ x_{i\ell}+x_{ip} \leq 0\big\}$$
and the usual bisubmodular polyhedron can be written as
\begin{equation}\label{eq:bisub}
 \big\{x \in \RR^{n \times \LLL}\ \big|\ \forall\  T \in \TTT^n\ \ \ \x(T) \leq f(T)\ \ \mbox{and}\ \ \forall\ i \in [n]\  \ \ x_{i\ell}+x_{ip} = 0\big\}.
\end{equation}

We now show how $U(f)$ can essentially be defined as a subset of $P(f)$. For that we need the notion of unified vectors, inspired by \cite{Kuivinen}.
\begin{definition}
For $k\geq2$ a vector $y \in \RR^{\LLL}$ is called \emph{unified}, if there exists a $\ell \in \LLL$ such that for all $p \in \LLL \setminus \{\ell\}$ one has $-y_p=y_\ell\geq 0$ and a vector $x \in \RR^{n \times \LLL}$ is called \emph{unified} if for all $i \in [n]$ the vector $x_i \in \RR^{\LLL}$ is unified.\footnote{As in \secref{proof} we can extend our notions to $k=1$ by calling a one-dimensional vector \emph{unified}, if it is in $\RR_{\leq 0}$.}
\end{definition}
For any $k$-submodular function $f: \TTT^n \ra \RR$ with $f(\0) = 0$, we define 
$$U(f) := \menge{x \in P(f)}{x \ \ \ \mbox{is unified}}.$$
This is a very natural embedding of $U(f)$ from \secref{properties} in the polyhedron $P(f)$.
In particular, for $k=2$ the set $U(f)$ is the usual bisubmodular polyhedron as in \eq{bisub}. For $k\geq 3$ the set $U(f)$ is not necessarily a polyhedron anymore. The subset of unified vectors of a similar polyhedron play an important role in the tractability result in \cite{Kuivinen}.

With this notation, our main result, \thmref{main}, reads
\begin{equation}\label{eq:Main2}
 \min_{T \in \TTT^n}f(T) = \max_{x \in U(f)} - \|x\|,
\end{equation}
where $\|x\|$ for $x \in U(f)$ is defined as in \secref{minmax} accordingly for the embedding. That is, for any $\ell_1, \dots, \ell_n \in \LLL$ we can write
$$\|x\| := \sum\limits_{i=1}^{n}|x_{i\ell_i}|.$$

Recently, in~\cite{FujishigeTanigawa}, a weaker version of \thmref{main} was obtained. The authors consider a different $k$-submodular polyhedron, which we will call $P_{FT}(f)$ to distinguish it from our $k$-submodular polyhedron $P(f)$. It is, in our notation, defined as
$$
P_{FT}(f) := \big\{x \in \RR^{n \times \LLL}\ \big|\ \forall\  T \in \TTT^n\ \ \ \x(T) \leq f(T)\big\}.
$$
\cite{FujishigeTanigawa} also introduces a different norm which is defined for $x \in \RR^{n \times \LLL}$ as
$$ \|x\|_{1,\infty} := \sum_{i=1}^{n}\max_{\ell_i \in \LLL}|x_{i\ell_i}|.$$
Theorem 3.1 in~\cite{FujishigeTanigawa} states that
$$\min_{T \in \TTT^n}f(T) = \max_{x \in P_{FT}(f)} - \|x\|_{1,\infty}.$$
We will show how this follows from \eqref{eq:Main2}. The inequality
$$\min_{T \in \TTT^n}f(T) \ge \max_{x \in P_{FT}(f)} - \|x\|_{1,\infty}$$
follows directly from the definitions. To prove
\begin{equation}\label{eq:FT}
\min_{T \in \TTT^n}f(T) \le \max_{x \in P_{FT}(f)} - \|x\|_{1,\infty}
\end{equation}
note that directly from the definitions we have
$$U(f)\subseteq P(f) \subseteq P_{FT}(f),$$
and
$$\|x\| = \|x\|_{1,\infty} \mbox{\ \ \ \ for all\  \ } x \in U(f).$$
It thus follows that
$$\min_{T \in \TTT^n}f(T) \stackrel{\eqref{eq:Main2}}{=} \max_{x \in U(f)} - \|x\| = \max_{x \in U(f)} - \|x\|_{1,\infty} \le \max_{x \in P_{FT}(f)} - \|x\|_{1,\infty}.$$
Inequality \eqref{eq:FT} and therefore Theorem 3.1 in~\cite{FujishigeTanigawa} follow.

In the remainder of the section, we collect some properties of $P(f)$.
\propref{supp} can be generalized to
\begin{prop}
For every $x \in P(f)$ the function $\x$ is $k$-supermodular.
\end{prop}
\begin{proof}
For every $i \in [n]$ the function $\x_i$ is $k$-supermodular by definition of $P(f)$, and this carries over to the sum.
\end{proof}

As in \secref{properties} we define for every $x \in P(f)$ the set
$$\FFF(x) := \menge{T \in \TTT^n }{\x(T) = f(T)}$$
of \emph{$x$-tight} elements and have
\begin{prop}\label{prop:Closed}
The set $\FFF(x)$ is closed under the operations $\sqcap$ and $\sqcup$, and the function $\x\big|_{\FFF(x)} = f\big|_{\FFF(x)}$ is $k$-modular.
\proofbox
\end{prop}
We define for every $x \in P(f)$ the set
$$\GGG(x) := \menge{(i,\{p,q\}) \in [n] \times {\LLL \choose 2}}{x_{ip}+x_{iq} = 0}.$$
In the following, we will investigate the vertices of $P(f)$.
An element $x \in P(f)$ is a vertex of $ P(f)$ if and only if it is the unique solution to the set of equations
\begin{eqnarray*}
\forall\ T \in \FFF(x)\ \ \ \ \ \ \ \ \ \ \x(T)&=& f(T)\\
\forall\ (i,\{p,q\}) \in \GGG(x)\ \ \ x_{ip}+x_{iq} &=&0
\end{eqnarray*}
A set $\BBB \subseteq \FFF(x) \cup \GGG(x)$ is called \emph{basis} for $x$ if $|\BBB|=kn$ and $x$ is the unique solution to the set of equations
\begin{eqnarray}
 \forall\ T \in  \BBB_1 := \BBB \cap \FFF(x)\ \ \ \ \ \ \ \ \ \ \ \x(T) &=& f(T)\label{bbb}\\
\forall\ (i,\{p,q\}) \in  \BBB_2 := \BBB \cap \GGG(x) \  \ \ \ x_{ip}+x_{iq} &=&0\nonumber
\end{eqnarray}

\begin{rem}\label{rem:1}
 If $k \leq 2$ we have $|\BBB_2|\leq (k-1)n$ and thus $|\BBB_1|\geq n$.
\end{rem}

For the sake of completeness we state the proof of the next lemma, which was Lemma 5 in \cite{ISCO}. We do not use it any further in this paper though.

\begin{lem}\label{lem:exchange}
For every $x \in P(f)$, every basis $\BBB$ for $x$, and all $S,T \in  \BBB_1$, there is a $\g \in \{T \sqcap S, T \sqcup S\} \cup \menge{(i,\{S_i,T_i\})}{i \in [n], \mbox{ $S_i$ and $T_i$ are different leaves}}$ such that replacing $T$ with $\g$ in $\BBB$ gives a basis for $x$, \ie $\big(\BBB \setminus \{T\}\big) \cup \{\g\}$ is a basis for $x$.
\end{lem}
\begin{proof}
Without changing the solution, we can do the following changes to \eqref{bbb}. 

By linear algebra, we can replace the equation $\x(T) = f(T)$ with
\begin{equation}\label{eq:+}
 \x(T)+ \x(S) = f(T)+f(S).
\end{equation}
By \propref{Closed} and the definition of $P(f)$ we have $T \sqcap S, T \sqcup S \in \FFF(x)$ and $(i,\{S_i,T_i\}) \in \GGG(x)$ for every $i \in [n]$ for which $S_i$ and $T_i$ are different leaves. We thus can replace \eq{+} by the equations
\begin{eqnarray*}
\x(T \sqcap S) &=& f(T \sqcap S)\\
\x(T \sqcup S) &=& f(T \sqcup S)\ \ \mbox{and}\\
x_{iS_i}+x_{iT_i} &=&0\ \ \ \mbox{for all $i \in [n]$ such that $S_i$ and $T_i$ are different leaves.}
\end{eqnarray*}
By linear algebra, one of them is enough.
\end{proof}

\section{Relation to Multimatroids}\label{sec:multimatroids}
In this section we discuss the relation between $k$-submodular functions and {\em multimatroids} introduced by Bouchet~\cite{Bouchet:I,Bouchet:II,Bouchet:III}.
First, we recall some definitions from~\cite{Bouchet:I,Bouchet:II,Bouchet:III} (adapted to our notation). 
We use the same definitions of sets $\TTT$, $\LLL=\TTT\setminus\{\oo\}$ and operations $\sqcup,\sqcap:\TTT\times\TTT\rightarrow\TTT$ as before.
We say that $T,U\in \TTT^n$ are 
\begin{itemize} 
\item {\em compatible} if $|\{T_i,U_i\}\setminus\{\oo\}|\le 1$ for all $i\in[n]$;
\item {\em $\bar i$-similar} for $i\in[n]$ if $T_j=U_j$ for all $j\in[n]\setminus \{i\}$.
\end{itemize}
\begin{definition}
A function $r:\TTT^n\rightarrow\mathbb Z_{\ge 0}$ is called a {\em rank function of a $k$-matroid} if it satisfies
\begin{enumerate}
\item $r({\bf 0})=0$.
\item If $T,U\in\TTT^n$ are $\bar i$-similar and $T_i=\oo$ then 
\begin{equation}
r(T)\le r(U) \le r(T)+1.
\end{equation}
\item If labelings $T,U\in\TTT^n$ are compatible then
\begin{equation}
r(T\sqcap U)+r(T\sqcup U)\le r(T)+r(U).
\end{equation}
\item If $T,U\in\TTT^n$ are $\bar i$-similar and $|\{T_i,U_i\}\setminus\{\oo\}|=2$ then
\begin{equation}
r(T\sqcap U)+r(T\sqcup U)\le r(T)+r(U)-1.
\end{equation}
\end{enumerate}
\end{definition}
Multimatroids are just a slight generaralization of $k$-matroids: the number of leaves is allowed to be different for
different $i\in[n]$.
We claim that a rank function of a $k$-matroid is a $k$-submodular function; this follows from
\begin{prop}
Function $f:\TTT^n\rightarrow \mathbb R$ is $k$-submodular if and only if it satisfies
\begin{eqnarray}
f(T\sqcap U)+f(T\sqcup U)\le r(T)+r(U) \label{eq:LIAHGAKJJFHAK}
\end{eqnarray}
(i) for all compatible $T,U\in\TTT^n$; and\\(ii) for all $\bar i$-similar $T,U\in\TTT^n$ with $|\{T_i,U_i\}\setminus\{\oo\}|=2$.
\end{prop}
\begin{proof}
The ``only if'' direction is trivial; let us consider the ``if'' part.
It is trivial for $k=1$, and for the case $k=2$ it was shown in~\cite{Ando:96}.
Suppose that $k\ge 3$. To show that~\eqref{eq:LIAHGAKJJFHAK} holds for arbitrary 
$T,U\in\TTT^n$, we can just restrict $f$ to a bisubmodular function
that uses leaves present in $T$ and $U$ and then apply the characterization of~\cite{Ando:96}
for the case $k=2$.
\end{proof}

\section{Discussion}\label{sec:disc}

In the submodular and bisubmodular cases, the Min-Max-Theorem led to polynomial time minimization algorithms.
In the case of the simple but non distributive lattice class called diamonds~\cite{Kuivinen}, it led to a pseudo-polynomial algorithm whose complexity depends
polynomially on the value of the function. 
Can we use the Min-Max-Theorem to design (pseudo-) polynomial algorithms for $k$-submodular functions?

Unfortunately, we still miss one important piece for designing such algorithms.
Name\-ly, we do not know at the moment whether the polyhedron $P(f)$ is \emph{well-characterized},
i.e.\ whether for each vector $x\in P(f)$ there is a certificate of the fact that $x\in P(f)$
which can be checked in polynomial time. Note, it suffices to have such certificates for vertices of $P(f)$,
since any vector $x\in P(f)$ can be represented as a convex combination of polynomially many vertices.

It is known that the good characterization property holds for the cases of submodular functions~\cite{FujiBook,Schrijver}, bisubmodular functions~\cite{Qi},
diamonds and modular lattices~\cite{Kuivinen}. 

To summarize, the existence of a polynomial time algorithm for $k$-submodu\-lar functions remains an open question, despite the Min-Max-Theorem.

\section*{Acknowledgements}
We would like to thank Andrei Krokhin for encouraging our cooperation, for helpful discussions, and for his critical reading of the manuscript.
We also thank Satoru Fujishige for pointing out the work of Bouchet on multimatroids~\cite{Bouchet:I,Bouchet:II,Bouchet:III} to us and Satoru Fujishige and Shin-ichi Tanigawa for finding a mistake in the preprint of this paper as outlined in the introduction.

\bibliographystyle{alpha}
\bibliography{lit}
\end{document}